\documentclass[conference]{IEEEtran}

\usepackage{cite}
\usepackage{graphicx}
\usepackage{amsfonts,amsmath,amssymb,bm}
\usepackage{tabularx}
\usepackage{xcolor}
\usepackage{subcaption}
\usepackage{booktabs}
\usepackage{comment}
\usepackage{caption}
\usepackage{amsthm}
\usepackage{algorithm}
\usepackage{multirow}
\usepackage{algpseudocode}
\usepackage{hyperref}
\usepackage{cleveref}

\crefname{figure}{Fig.}{Figs.}
\Crefname{figure}{Fig.}{Figs.}

\definecolor{burgundy}{HTML}{800020}
\definecolor{normandy}{RGB}{85, 107, 47}

\DeclareMathOperator*{\argmin}{arg\,min}

\newtheorem{theorem}{Theorem}
\newtheorem{lemma}{Lemma}
\newtheorem{corollary}{Corollary}

\theoremstyle{definition}
\newtheorem{definition}{Definition}
\newtheorem{remark}{Remark}
\newtheorem{example}{Example}

\algrenewcommand\algorithmicindent{1.0em}
\algrenewcommand\algorithmicrequire{\textbf{Input:}}
\algrenewcommand\algorithmicensure{\textbf{Output:}}

\hypersetup{
    colorlinks=true,
    linkcolor=black,
    citecolor=black,
    urlcolor=blue,
    breaklinks=true,
    bookmarksdepth=3,
    pdfauthor={Noor Khial, Naram Mhaisen, Loay Ismail, Amr Mohamed},
    pdftitle={Between Scene Noise and Drift: Optimistic Online Policy Selection for UAV Radar Search},
    pdfsubject={Online Learning, UAV Systems},
    pdfkeywords={UAV search, online convex optimization, adaptive regret}
}

\begin{document}

\title{SEArch: Optimistic Policy Selection Between Scene Noise and Drift for UAV Radar Search}

\author{%
\IEEEauthorblockN{Noor Khial\IEEEauthorrefmark{1},
Naram Mhaisen\IEEEauthorrefmark{2},
Loay Ismail\IEEEauthorrefmark{1},
Amr Mohamed\IEEEauthorrefmark{1}}
\IEEEauthorblockA{%
\IEEEauthorrefmark{1}College of Engineering, Qatar University, Qatar\\
\IEEEauthorrefmark{2}Faculty of Electrical Engineering, Mathematics, and Computer Science, Delft University of Technology, The Netherlands\\
Email: nk1703044@qu.edu.qa, N.Mhaisen@tudelft.nl, \{loay.ismail, amrm\}@qu.edu.qa}%
}

\maketitle

\begin{abstract}
Unmanned Aerial Vehicles (UAVs) equipped with radar sensors are deployed for target search missions in diverse environments, where targets exhibit characteristic signatures (e.g., respiration micro-motion in human search) detectable through occlusions. A fundamental challenge arises from shifts in radar statistics as the UAV moves through a dynamic and potentially non-stationary environment, rendering any fixed signal-processing strategy suboptimal; yet perception and adaptation must run onboard a resource-constrained aerial node in real time. Since no single detector performs well across all conditions, we adopt a multi-policy paradigm and formulate UAV target search as an online policy selection problem over a library of specialized detectors, with performance measured by regret, the cumulative loss gap relative to the best policy in each scene. The setting couples in-scene stochastic noise with inter-scene shifts. Whereas prior methods capture only one regime, we account for both through the Stochastically Extended Adversary (SEA) framework, without requiring oracle knowledge of scene dynamics. Because adaptation must run at the UAV, we instantiate SEA through \textsc{SEArch}, a lightweight optimistic Follow the Regularized Leader (OFTRL) selector with an adaptive learning rate, achieving regret $O(\bar{\sigma}_T \sqrt{T} + \sqrt{J})$, where $\bar{\sigma}_T$ captures radar measurement noise and $J$ is the number of scene transitions over the mission horizon $T$. To enable rapid adaptation under frequent scene changes, we further introduce \textsc{W-SEArch}, a windowed variant that restarts every $w$ rounds and achieves regret $O(\bar{\sigma}_I \sqrt{w})$ under at most one transition per window. Experiments show up to 30\% regret reduction compared to non-adaptive baselines across a range of non-stationary settings.
\end{abstract}
\begin{IEEEkeywords}
UAV search, Radar sensing, Non-stationary environments, Policy selection, Online learning
\end{IEEEkeywords}

\section{Introduction}
\label{sec:introduction}

Unmanned Aerial Vehicles (UAVs) equipped with radar sensors are deployed for target search in maritime rescue, disaster response, and terrestrial operations, where targets exhibit characteristic signatures, such as respiration micro-motion in human search~\cite{rohman2021throughwall,jing2022respiration,khial2025multi}, detectable under occlusions. Yet detection operates at the limits of radar sensitivity: the signature is weak, clutter is strong, and propagation losses are severe. It is further strained as radar statistics change with UAV motion through a dynamic and potentially non-stationary environment, leaving even small changes (e.g., wall materials) enough to sharply degrade performance. Compounding this, perception and adaptation must run onboard a low-altitude UAV acting as the bottom tier of a low-altitude IoT network~\cite{abdellatif2025pdsr}, on a resource-constrained node where cloud offloading is precluded by the ultra-low-latency and connectivity demands of search missions. Any adaptive intelligence must therefore stay lightweight in computation and sensing.

\begin{figure}[t!]
    \centering
    \includegraphics[width=\columnwidth]{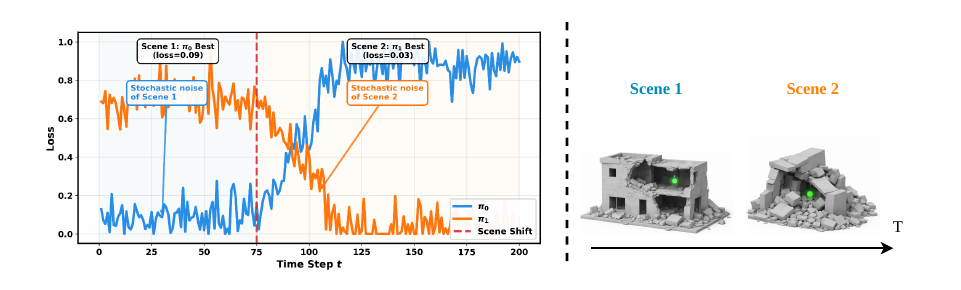}
    \caption{\textbf{Scene Drift:} $\pi_0$ is optimal under Scene 1's stochastic noise and $\pi_1$ is optimal in Scene 2. Our approach aims to place high weight on $\pi_0$ during Scene 1 without oscillating due to noise, then adapts to $\pi_1$ after the transition, without prior knowledge of the shift timing or the noise scale.}
    \label{fig:two_policy_scene_shift}
\end{figure}

No single signal processing strategy performs well across this range of conditions, since a detector tuned for one geometry or clutter profile degrades once the environment shifts. We therefore adopt a multi policy paradigm. Rather than designing one processor for all scenarios, we assume a finite library of $K$ policies $\Pi=\{\pi_1,\dots,\pi_K\}$, each trained offline for a particular scene archetype\footnote{We refer the reader to \cite{khial2025drone} which covers one example of a pre-trained set. In this work we mainly focus on the online selection layer.} and possibly encapsulating combinations of detection thresholds or clutter conditions. The task then becomes one of online selection, in which at each step $t$ an online selector chooses a probability distribution $\bm{x}_t \in \Delta^K$ over $\Pi$, executes a policy according to $\bm{x}_t$, and incurs a bounded loss $\ell_t(\bm{x}_t)$ reflecting detection performance, with the goal of tracking the policy of minimum loss in the current scene without a priori knowledge of scene drifts or noise magnitude. We cast this as Online Convex Optimization (OCO)~\cite{hazan2016oco}, in which performance is measured by \emph{regret}, the cumulative loss gap relative to the best policy in each scene, and a sublinear regret ensures the selector eventually matches an oracle with perfect scene knowledge.

The difficulty of detection lies in the presence of two distinct sources of uncertainty acting on the selector. The first, \emph{in-scene stochastic noise}, arises because each scene $s \in \mathcal{S}$ induces a distribution $\mathcal{N}(\mu_s, \sigma_s^2)$ over radar measurements, where $\sigma_s^2$ captures fluctuations from UAV micro-motions and clutter. The second, \emph{inter-scene distributional shift}, arises as the environment evolves and the measurement distribution shifts from $\mathcal{N}(\mu_s, \sigma_s^2)$ to $\mathcal{N}(\mu_{s'}, \sigma_{s'}^2)$, with both the mean gap $\|\mu_s - \mu_{s'}\|$ and the transition timing unknown and arbitrary (see \Cref{fig:two_policy_scene_shift}). The selector must therefore be stable under noise, so that transient fluctuations do not trigger oscillation between policies, yet responsive at transitions, so that scene changes are tracked. Most existing approaches address only one of these regimes. Adversarial OCO methods~\cite{khial2025eswa, 9165951} treat all variation as arbitrary and yield conservative bounds independent of the actual difficulty, whereas stochastic models cannot capture sudden changes in scenes. To account for both, we adopt the \emph{Stochastically Extended Adversary} (SEA) framework~\cite{sachs2022between}, which splits the problem difficulty into two measurable parts, namely the noise accumulated within scenes and the variation accumulated across them, denoted $\bar{\sigma}_T^2$ and $\bar{\Sigma}_T^2$, respectively. The resulting bounds interpolate between the stable and non-stationary regimes rather than committing to either.

Because perception and adaptation must run at the UAV, we instantiate SEA through \textsc{SEArch}, an optimistic Follow the Regularized Leader (OFTRL)~\cite{rakhlin2013predictable} selector with an adaptive learning rate~\cite{sachs2022between}. OFTRL is lightweight, since each round reduces to a simplex projection over the $K$ policies, with no online training and no accelerator, making it well suited to a constrained onboard node. The optimistic component further sharpens this advantage. In each round the learner predicts the next loss from the previous radar observation, which proves accurate while the UAV dwells in a scene, whereas the adaptive learning rate automatically grows cautious when predictions fail at a transition. The resulting bound tightens to near optimal stochastic rates within a single scene and relaxes to non-stationary rates when changes are frequent. We further refine this bound by noting that under piecewise-stationary scenes the adversarial variation concentrates at transitions; hence $\bar{\Sigma}_T$ scales as $O(\sqrt{J/T})$, and the bound becomes $O(\bar{\sigma}_T \sqrt{T} + \sqrt{J})$. Thus, the switching penalty grows only as $\sqrt{J}$ in the number of scene transitions.

Standard FTRL methods aggregate the entire observation history into a single state vector~\cite[Section~5.2]{hazan2016oco}, which induces inertia after scene changes~\cite{mhaisen25a, Ahn2024UnderstandingAO}. Concretely, when the UAV leaves a hallway for a room, the clutter spectrum shifts, yet the pre-transition frames still dominate the update, so the selector keeps favoring the previous policy. This inertia becomes particularly severe in UAV missions where scene durations vary unpredictably with trajectory and speed. To overcome the history effects, we introduce \textsc{W-SEArch}, a windowed variant that restricts updates to the most recent $w$ radar observations. Motivated by empirical evidence that UAV radar coherence persists only over short horizons~\cite{jing2022respiration,rohman2021throughwall}, \textsc{W-SEArch} prunes obsolete measurements from previously traversed scenes. Namely, at most one scene boundary falls within any window of length $w$, we establish an adaptive regret of $O(\bar{\sigma}_I \sqrt{w})$ on any contiguous interval of that length.

Our contributions are summarized as follows.
\begin{itemize}
  \item We formulate UAV search as online policy selection over a finite library of signal processing policies, and model the uncertainty structure (in-scene noise and inter-scene shifts) using the SEA framework, which yields problem-dependent regret bounds that scale with the actual environmental difficulty.

  \item We apply \textsc{SEArch}, an OFTRL selector with an adaptive learning rate, achieving a regret bound of $O(\bar{\sigma}_T \sqrt{T} + \sqrt{J})$ that decomposes into stochastic and switching components, where $J$ is the number of scene transitions and the switching penalty grows only as $\sqrt{J}$.

  \item We introduce \textsc{W-SEArch}, a windowed OFTRL variant that retains only the most recent $w$ observations, enabling continuous adaptation to irregular scene changes with regret $O(\bar{\sigma}_I \sqrt{w})$ on any contiguous interval of length $w$. The restart mechanism keeps the selector lightweight for onboard deployment.

  \item We validate our approach across various levels of non-stationarity, achieving low regret and accurate scene tracking with up to 30\% regret reduction compared to cumulative-history baselines.
\end{itemize}
\section{Related Work}
\label{sec:related}

\textbf{UAV Radar Sensing.}
UAV-mounted sensor systems have emerged as a key technology for target detection across diverse environments~\cite{rohman2021throughwall}. Multiple sensing modalities have been explored for target finding missions, including vision-based systems~\cite{shurrab2023rl,hussain2023predictive}, thermal imaging for detecting body heat signatures~\cite{rudol2008human,andriluka2010monocular}, and radar systems capable of penetrating debris and walls~\cite{rohman2021throughwall, qassmi2026obstacle, jing2022respiration}. These studies demonstrate the feasibility of target sensing from a hovering UAV.

A key challenge is the non-stationary nature of sensor returns. Empirical evidence shows that signal statistical properties remain stable only over short windows before platform motion, environmental changes, and multipath variations induce significant drift~\cite{jing2022respiration,hussain2023predictive}. Recent work has explored integrating adaptive sensing with UAV path planning~\cite{shao2022intelligent,khial2025drone} and learning-based approaches for target localization in noisy sensor environments~\cite{mohammed2023deep,shurrab2023rl}. However, these approaches assume fixed signal processing pipelines that do not adapt to changing operating conditions across different search scenes. Our work addresses this gap by formulating adaptive policy selection that accounts for the non-stationary nature of UAV sensor returns across different environments.

\textbf{UAV-Based Search.}
The search for UAV has been studied from multiple perspectives for target finding missions. Classical approaches focus on coverage path planning~\cite{cabreira2019survey,shao2020efficient}, where the objective is to cover a search region while minimizing mission time or energy consumption. Information-theoretic methods~\cite{julian2012distributed} formulate search as maximizing information gain about target location, using Bayesian filters to maintain belief distributions over possible target positions. Optimal control approaches~\cite{ragi2016uav} cast the search as a trajectory optimization problem with detection probability models. More recent work has incorporated learning methods for UAV search~\cite{kurunathan2023machine}: deep reinforcement learning has been applied to learn search policies from simulated missions~\cite{shurrab2023rl, kurunathan2023machine}, while multi-agent coordination for cooperative search has been explored through consensus-based approaches~\cite{choi2009consensus,sheng2006distributed}. Active Vision-based perception systems~\cite{krause2008near} demonstrate adaptive sensing strategies but assume stationary environments.

A growing line of work frames search as policy selection from finite libraries. Motion primitive sensing with learned selectors~\cite{goel2022active} enables rapid adaptation while maintaining interpretability, and bandit-based multi-robot target tracking~\cite{xu2023bandit} addresses exploration-exploitation tradeoffs over discrete actions. Recent UAV search work has also started to incorporate online learning:~\cite{khial2025eswa} formulates UAV search as an adversarial bandit problem over finite strategy sets, and~\cite{khial2025drone} extends this to reinforcement learning policies. However, both assume a purely adversarial environment, yielding guarantees that do not exploit the structure of scene-based environments.

\textbf{Online Learning.}
OCO provides a principled framework for sequential decision making under time-varying conditions~\cite{hazan2016oco}: at each round, a learner updates its strategy based on the observed feedback, with the goal of minimizing cumulative loss over the horizon. This framework has found a wide range of applications~\cite{khial2025drone, mhaisen2022online}. Performance is analyzed through two regret notions. \emph{Static regret} compares with the single best fixed policy in hindsight, whereas \emph{adaptive regret}~\cite[Section 15.2]{orabona2019modern} competes with the best policy on every contiguous sub-interval of the horizon $T$. The latter is the appropriate metric for our setting, as the optimal policy varies across scenes rather than remaining fixed.

Several frameworks model non-stationary environments with varying degrees of structure. Adversarial models~\cite[Ch.~11]{lattimore2020bandit, hoi2021online} permit arbitrary changes but yield conservative $O(\sqrt{T})$ guarantees that are independent of the actual degree of non-stationarity. At the other extreme, stochastic multi-armed bandits~\cite[Ch.~4]{lattimore2020bandit, hoi2021online} assume i.i.d. rewards, which fails to capture settings where distributions shift over time. The SEA framework~\cite{sachs2022between, wang2025parameter} interpolates between these two extremes by decomposing problem difficulty into two explicit components: cumulative stochastic variance, which captures measurement noise within each regime, and cumulative adversarial variation, which captures distributional shifts across regimes. This decomposition yields problem-dependent bounds that scale with the actual environmental difficulty, tightening to near-optimal stochastic rates in stationary environments and relaxing to adversarial rates when shifts are frequent.

\textbf{Non-Stationary Learning.}
Practical algorithms for non-stationary learning employ  forgetting mechanisms. Sliding-window methods~\cite{gaillard2014second,luo2015achieving} base decisions on recent observations but require careful tuning of the window size. Restart methods~\cite{jadbabaie2015online} periodically reset the learning state, either at fixed intervals or when triggered by detected changes. Meta-learning approaches maintain multiple learners and weight them based on observed performance~\cite{finn2019online, al2017continuous, zhang2018adaptive}, and~\cite{mhaisen2026partially} balances between greedy and lazy updates via partial laziness. 

FTRL~\cite[Ch.~7]{orabona2019modern} and its optimistic variant~\cite{rakhlin2013predictable} are core online learning tools that use predictions of future losses to accelerate convergence in predictable environments. However, standard FTRL aggregates the entire loss history into a cumulative state, which hinders adaptation when the environment changes~\cite{mhaisen25a}. Recent analysis reveals that the bottleneck is the unbounded growth of historical information decoupled from current iterates~\cite{mhaisen2026partially, mhaisen25a}. This motivates history-pruning variants that discard obsolete observations while maintaining the benefits of optimistic predictions.

\textbf{Contributions.}
Our work on UAV search~\cite{khial2025eswa, khial2025drone} adopts an adversarial view, yielding horizon dependent regret bounds that ignore the underlying scene structure. The present work departs from this view by modeling UAV missions as a sequence of scene induced stationary segments analyzed through the SEA framework, producing problem dependent regret bounds that scale with scene variance $\bar{\sigma}_T^2$ and the number of transitions $J$. A windowed OFTRL variant further improves adaptability in rapid missions, with regret controlled by the number of transitions per window. The contribution is also modular with respect to the policy library; the offline pool is constructed in~\cite{khial2025drone}, while here we target the orthogonal online selection layer that can sit atop any such pool.
\section{System Model and Problem Formulation}
\label{sec:system_problem}

We consider a UAV equipped with a radar conducting a search mission in a \emph{dynamic, potentially non-stationary} environment. We position the UAV as the bottom tier of a low-altitude IoT network, consistent with the multi-tier search architecture of~\cite{abdellatif2025pdsr}, in which low-altitude UAVs (LA-UAVs) operate close to the ground to carry out the search operation. Our work targets the policy selection layer that runs onboard the LA-UAV during this sensing mission: as the LA-UAV navigates the environment, its radar exploits characteristic target signatures (e.g., respiration-induced micro-motion in human search~\cite{rohman2021throughwall}, with all processing performed locally.

\subsection{Motivating Examples: Mission Regimes}
\label{sec:regime_examples}

We illustrate three mission behaviors that motivate our approach. Each is characterized by the number and frequency of distinct operating =regimes the UAV encounters.

\begin{example}[Single-Regime Dominance]
\label{ex:single_regime}
When the UAV dwells in a region with nearly fixed geometry and propagation conditions (e.g., a single facade segment in terrestrial search, or a uniform sea state in maritime operations), the radar operating environment remains constant over extended intervals.
\end{example}

\begin{example}[Gradual Drift]
\label{ex:gradual_drift}
Many missions exhibit \emph{slowly evolving regimes}, where continuous range and aspect variation cause the radar environment to change smoothly. The optimal policy may then shift gradually over time.
\end{example}

\begin{example}[Rapid Multi-Regime Switching]
\label{ex:rapid_switching}
When the UAV repeatedly encounters distinct operating conditions over short horizons (e.g., alternating between open areas and heavily occluded zones, or transitioning between varying wall materials in terrestrial settings or changing sea states in maritime operations), the optimal policy changes frequently and abruptly.
\end{example}

\subsection{Multi-Policy Paradigm}

We adopt a multi policy approach and assume access to a library of $K$ specialized signal processing policies, each optimized for a particular radar operating regime. The learning task is to adaptively select among these policies online as the operating regime evolves. This selection problem addresses two sources of uncertainty: (i)~\emph{in-scene stochastic noise} from measurement variability and clutter within a fixed scene, and (ii)~\emph{inter-scene non-stationarity} from changes in the underlying radar return distribution as scenes transition.

We adopt the SEA framework~\cite{sachs2022between}, which models environments that combine stochastic fluctuations within regimes and adversarial shifts across regimes. 

\subsubsection{Notation}
Throughout this paper, we denote scalars in italic (e.g., $T$, $K$), vectors in bold lowercase (e.g., $\bm{x}_t\in\mathbb{R}^K$, $\bm{\ell}_t\in\mathbb{R}^K$), and sets in calligraphic font (e.g., $\mathcal{S}$, $\Delta^K$). We use $[T] := \{1,\dots,T\}$ to denote integer intervals, $\langle\cdot,\cdot\rangle$ for the Euclidean inner product, and $\|\cdot\|$ for the Euclidean ($\ell_2$) norm.

\subsection{Scene Model and Piecewise Stationarity}
\label{sec:scene_model}

Each radar operating condition is modeled as a \emph{scene}. Let $\mathcal{S} = \{s_1, s_2, \ldots, s_M\}$ denote a finite set of possible scenes, where each $s \in \mathcal{S}$ represents a characteristic physical configuration (environmental conditions, geometry, target location, and propagation environment) that induces a measurement distribution $D_s$ over radar returns. At each time step $t \in [T]$, the UAV operates in scene $s_t \in \mathcal{S}$, which determines the statistical properties of radar measurements.

UAV-mounted radar observations remain statistically consistent only over short horizons before drift or scene changes occur~\cite{jing2022respiration,rohman2021throughwall}. We formalize this behavior via a \emph{piecewise-stationary model}: the horizon $[T]$ is partitioned into $J$ contiguous stationary segments, each corresponding to a fixed scene.

\begin{definition}[Stationary Segments]
\label{def:segments}
Let $1 = \tau^{(1)} < \tau^{(2)} < \cdots < \tau^{(J)} < \tau^{(J+1)} = T+1$ denote the scene transition times (unknown to the learner). The horizon $[T]$ is partitioned into $J$ contiguous segments $I^{(j)} = [\tau^{(j)}, \tau^{(j+1)}-1]$ for $j=1,\ldots,J$. Within each segment $I^{(j)}$, the scene remains fixed at some $s^{(j)} \in \mathcal{S}$, inducing a stationary measurement distribution $D_{s^{(j)}}$. Equivalently, $s_t = s^{(j)}$ for all $t \in I^{(j)}$. At each boundary $\tau^{(j)}$ for $j \geq 2$, the scene switches and the distribution changes from $D_{s^{(j-1)}}$ to $D_{s^{(j)}}$.
\end{definition}

\subsubsection{Radar Measurement Distribution}
Let $\xi_t$ denote the processed radar measurement at round $t$ (e.g., a feature vector extracted from raw returns). Each $s \in \mathcal{S}$ induces a probability distribution $D_s$ over the measurement space, so that
\begin{equation}
\label{eq:measurement_dist}
\xi_t \sim D_{s_t}.
\end{equation}
We assume that measurements exhibit Gaussian noise around a scene-specific mean, i.e., $D_s = \mathcal{N}(\mu_s, \sigma_s^2)$, where observations fluctuate around the mean $\mu_s$ due to \emph{in-scene stochastic noise} with variance $\sigma_s^2$.

\begin{figure*}[t!]
    \centering
    \includegraphics[width=\textwidth]{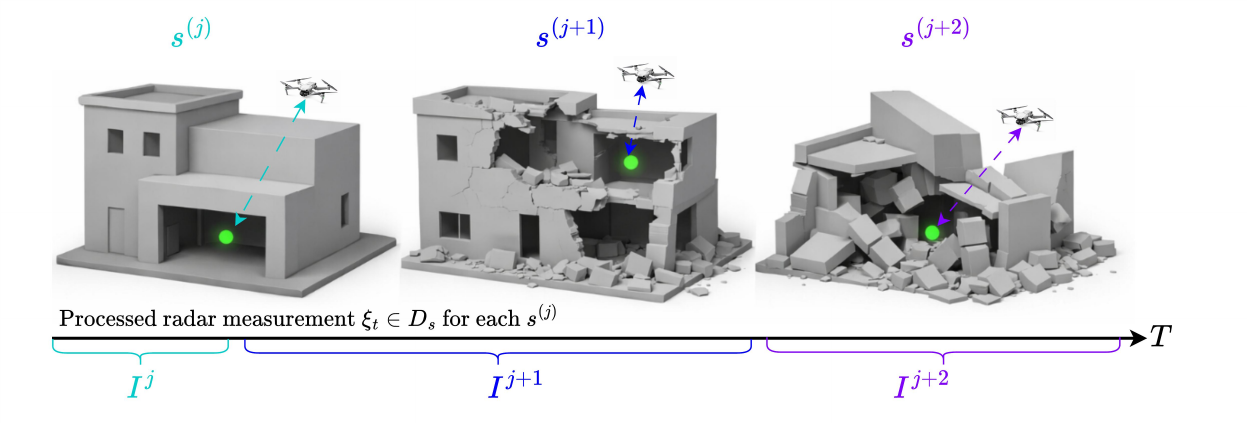}
    \caption{Example of the scene sequence model. The mission evolves through a sequence of scenes $s^{(j)}\in\mathcal{S}$, each corresponding to a characteristic radar operating regime (e.g., environmental conditions, geometry, and target location) that persists over a stationary segment $I^{(j)} = [\tau^{(j)}, \tau^{(j+1)}-1]$. The green circle denotes the hidden target (e.g., a human survivor) that the UAV seeks to detect through partial occlusions. Within each segment, the scene induces a measurement distribution $D_{s^{(j)}}$, from which radar observations $\xi_t$ are drawn (for instance, micro-Doppler feature vectors extracted from raw radar returns). The goal is to detect the presence of the target by selecting, online, the signal processing policy best suited to the current scene statistics.}
    \label{fig:sys_model}
\end{figure*}

\subsection{Signal Processing Policies and Online Selection}
\label{sec:policy_model}

\subsubsection{Policy Library}
We assume a finite library of $K$ signal processing policies $\Pi = \{\pi_1,\dots,\pi_K\}$, each optimized offline for a particular scene archetype using labeled radar data or domain heuristics. Each policy $\pi_i$ is a detector that processes the radar measurement $\xi$ to decide on the presence or absence of the target, and is associated with a loss function $L(\pi_i, \xi) \in [0,1]$ that captures its detection error (e.g., missed detections or false alarms) under measurement $\xi$. 

The offline construction of the policy library is complementary to, and independent of, the online selection problem we address here. Our approach is therefore \emph{modular}: any suitable library can be plugged in at the bottom of the stack, while the online selector on top remains unchanged. For one concrete instantiation of such a library for UAV-based human search, including reinforcement-learning-based policy primitives, we refer the reader to~\cite{khial2025drone}. In this work, we focus on the upper layer: given any such library, select online, at each time step $t$, the policy from $\Pi$ that best matches the current scene statistics, without prior knowledge of transitions or their timing.

\subsubsection{Online Selection Protocol}

At each round $t \in [T]$, the learner selects a weight vector $\bm{x}_t\in\Delta^K$ over policies, where $\Delta^K := \{ \bm{x} \in \mathbb{R}_{\ge 0}^K : \sum_{i=1}^K x_i = 1 \}$ is the probability simplex. After the selection, the environment determines the current scene $s_t$, a sample $\xi_t \sim D_{s_t}$ is observed, and the losses of all policies are revealed, yielding the full loss vector $\bm{\ell}_t = (L(\pi_1,\xi_t),\dots,L(\pi_K,\xi_t)) \in [0,1]^K$.

\subsubsection{Weighted Loss}
The loss induced by the $\bm{x}_t$ is
\begin{equation}
\label{eq:weighted_loss}
\ell_t(\bm{x}_t) = \langle\bm{x}_t,\bm{\ell}_t\rangle
= \sum_{i=1}^K x_{t,i}\,L(\pi_i,\xi_t).
\end{equation}

\begin{remark}[Assumptions]
\label{rem:assumptions}
Our formulation rests on:
\begin{itemize}
\item[(A1)] \textbf{Full information.} We observe the loss vector $\bm{\ell}_t$ for all $K$ policies at every round. This is a structural property of the sensing pipeline. Every policy $\pi_i$ is a signal processing detector applied to the same received radar return $\xi_t$, hence the full vector $\bm{\ell}_t = (L(\pi_1,\xi_t),\dots,L(\pi_K,\xi_t))$ is recovered from a single measurement, acquired once per round. Thus full feedback incurs no additional sensing  and no extra radar acquisitions, since the selector reuses the detection computation already performed onboard.
\item[(A2)] \textbf{Piecewise stationarity.} The distribution $D_{s_t}$ is constant within segments and changes only at segment boundaries, as formalized in \Cref{def:segments}.
\item[(A3)] \textbf{Unknown switches.} The algorithm does not require knowledge of $J$ or of the transition times $\{\tau^{(j)}\}$ and adapts purely from observed losses.
\end{itemize}
\end{remark}

\begin{figure*}[t!]
    \centering
    \includegraphics[width=\textwidth]{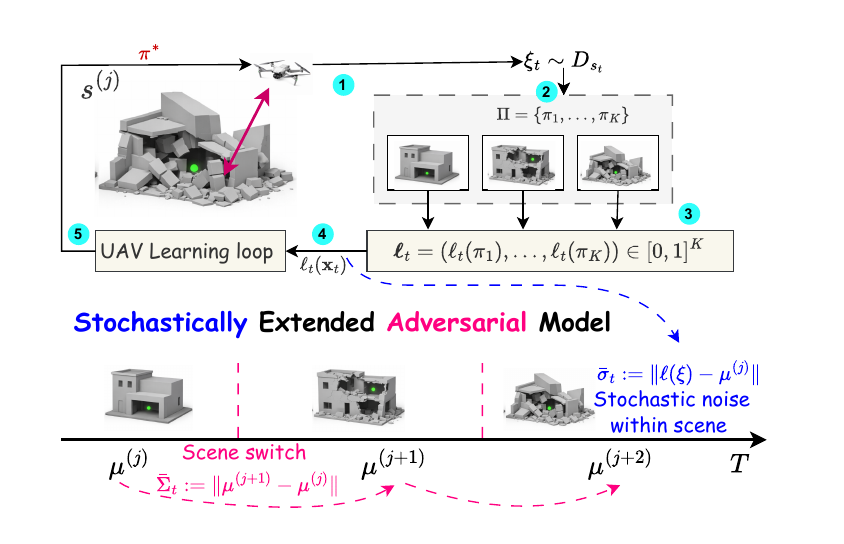}
    \caption{SEA mapping of online policy selection under scene-dependent radar noise and switching. At round $t$, the UAV operates in a single active scene $s_t$, which induces a piecewise-stationary measurement distribution $D_{s_t}$. A radar measurement $\xi_t \sim D_{s_t}$ is collected and processed simultaneously by all $K$ policies in $\Pi$, yielding the loss vector $\bm{\ell}(\xi_t) = (L(\pi_1, \xi_t), \dots, L(\pi_K, \xi_t))$ with mean $\bm{\mu}_t=\mathbb{E}_{\xi\sim D_{s_t}}[\bm{\ell}(\xi)]$. SEA measures both the in-scene stochastic variance, which captures measurement noise within the active scene, and the inter-scene drift, which captures the shift in loss statistics as the environment transitions across scenes.}
    \label{fig:policy_selection}
\end{figure*}

\subsection{Problem Formulation}
\label{sec:problem_formulation}

\subsubsection{Objective}
The challenge is to adapt online to the current scene under both scene noise and transitions. We therefore seek a sequence of weight vectors $(\bm{x}_t)_{t=1}^T\subset\Delta^K$ that minimizes the cumulative loss $\sum_{t=1}^T \ell_t(\bm{x}_t)$.

\subsubsection{Adaptive Regret}
We measure performance relative to the best fixed policy on each segment. By \Cref{def:segments}, the mission horizon $[T]$ is partitioned into segments $I^{(j)} = [\tau^{(j)}, \tau^{(j+1)}-1]$. We define the regret on segment $I^{(j)}$ as
\begin{equation}
\label{eq:segment_regret}
\mathrm{Reg}_{I^{(j)}}
:= \sum_{t\in I^{(j)}} \ell_t(\bm{x}_t)
\;-\;
\min_{\bm{x}\in\Delta^K}\sum_{t\in I^{(j)}}\ell_t(\bm{x}).
\end{equation}
Since $\ell_t(\bm{x})=\langle\bm{x},\bm{\ell}_t\rangle$ is linear and $\Delta^K$ is the simplex, the optimal weight on $I^{(j)}$ places all mass on a single policy:
\[
\min_{\bm{x}\in\Delta^K}\sum_{t\in I^{(j)}}\ell_t(\bm{x}) = \min_{i\in[K]}\sum_{t\in I^{(j)}}L(\pi_i,\xi_t).
\]
More generally, for missions with regime changes we aim to compete with the best policy over \emph{any} contiguous window of length $w$. This leads to the \emph{adaptive regret}~\cite[Section 15.2]{orabona2019modern}:
\begin{equation}
\label{eq:adaptive_regret_general}
\begin{aligned}
&\mathrm{Reg}_T(w) := \max_{[r,r+w-1] \subseteq [T]} \\
&\Bigg( \sum_{t=r}^{r+w-1} \ell_t(\bm{x}_t) 
\quad - \min_{\bm{x} \in \Delta^K}\sum_{t=r}^{r+w-1} \ell_t(\bm{x})\Bigg).
\end{aligned}
\end{equation}

\subsection{Stochastically Extended Adversary (SEA) Framework}
\label{sec:sea_mapping}

Our radar setting exhibits measurement noise within scenes and abrupt transitions between scenes. The SEA framework~\cite{sachs2022between} decomposes problem difficulty into a cumulative stochastic variance term and an inter-scene drift term.

\subsubsection{Mapping to SEA}
Recall that at round $t$ the learner selects $\bm{x}_t\in\Delta^K$ and observes a loss vector $\bm{\ell}_t\in[0,1]^K$, incurring the weighted loss $\ell_t(\bm{x}_t)$ in \Cref{eq:weighted_loss}. From \Cref{eq:measurement_dist} and \Cref{def:segments}, the measurement sample is generated by a scene-induced distribution $\xi_t \sim D_{s_t}$ that is stationary over segments with unknown boundaries. The SEA model formalizes this by allowing the environment to choose the sequence of distributions $(D_{s_t})_{t=1}^T$ (capturing inter-scene changes), while each $\xi_t$ is an observed sample from $D_{s_t}$.

\subsubsection{Loss Function}
To connect \Cref{eq:weighted_loss} to the stochastic sampling in \Cref{eq:measurement_dist}, we define the vector $\bm{\ell}(\xi) := (L(\pi_1,\xi),\ldots,L(\pi_K,\xi))\in[0,1]^K$ and the SEA loss
\[
f(\bm{x},\xi) := \langle\bm{x},\bm{\ell}(\xi)\rangle.
\]
With $\xi_t\sim D_{s_t}$, the observed loss satisfies $\bm{\ell}_t=\bm{\ell}(\xi_t)$, and the incurred loss coincides with the SEA loss: $\ell_t(\bm{x}_t) = f(\bm{x}_t,\xi_t)$.

SEA analyzes performance through the expected loss $F_t(\bm{x}) := \mathbb{E}_{\xi\sim D_{s_t}}[f(\bm{x},\xi)]$. With the mean loss vector
\[
\bm{\mu}_t := \mathbb{E}_{\xi\sim D_{s_t}}[\bm{\ell}(\xi)]\in[0,1]^K,
\]
we have $F_t(\bm{x})=\langle\bm{x},\bm{\mu}_t\rangle$. Under piecewise stationarity (\Cref{def:segments}), the mean loss is constant within segments: for $t \in I^{(j)}$, $\bm{\mu}_t = \bm{\mu}^{(j)}$, where $\bm{\mu}^{(j)} \in [0,1]^K$ denotes the time-invariant mean loss vector for segment $j$. Equivalently,
\begin{equation}
\label{eq:mu_segment}
\bm{\mu}^{(j)} := \mathbb{E}_{\xi \sim D_{s^{(j)}}}[\bm{\ell}(\xi)]
= \bm{\mu}_t \quad \text{for all } t \in I^{(j)}.
\end{equation}
Thus $\bm{\mu}_t$ changes only at the switch times $\{\tau^{(j)}\}_{j=2}^{J}$.

\begin{definition}[SEA Stochastic Variance and Inter-Scene Drift]
\label{def:variance_variation}
SEA quantifies problem difficulty through two measures:
\begin{itemize}
\item \textbf{Cumulative stochastic variance} (in-scene noise):
\begin{equation}
\label{eq:sea_variance}
\bar{\sigma}_T^2
:= \frac{1}{T}\sum_{t=1}^T \mathbb{E}_{\xi\sim D_{s_t}}\!\left[\|\bm{\ell}(\xi) - \bm{\mu}_t\|^2\right].
\end{equation}
This measures the average dispersion of policy losses within each scene-induced distribution, capturing measurement noise and clutter.

\item \textbf{Cumulative inter-scene drift}:
\begin{equation}
\label{eq:sea_variation_energy}
\bar{\Sigma}_T^2
:= \frac{1}{T}\sum_{t=2}^T \|\bm{\mu}_t-\bm{\mu}_{t-1}\|^2.
\end{equation}
This measures the average shift in the mean loss vector.
\end{itemize}
\end{definition}

\subsubsection{Interpretation for Scene Changes}
Under piecewise stationarity (\Cref{def:segments}), $\bm{\mu}_t = \bm{\mu}^{(j)}$ is constant within each segment $I^{(j)}$ and changes only at switch times. Thus, $\bar{\Sigma}_T^2$ concentrates its mass on the $J-1$ scene boundaries:
\[
\bar{\Sigma}_T^2 = \frac{1}{T} \sum_{j=2}^J \|\bm{\mu}^{(j)} - \bm{\mu}^{(j-1)}\|^2.
\]
This makes explicit that drift occurs only at the transitions $\{\tau^{(2)}, \ldots, \tau^{(J)}\}$ and remains zero within segments.

Beyond the piecewise-stationary case, $\bar{\Sigma}_T^2$ is defined per time step and therefore captures gradual drift as a worst-case: when drift is smooth rather than concentrated at switches, $\bar{\Sigma}_T^2$ accumulates small per-step contributions $\|\bm{\mu}_t - \bm{\mu}_{t-1}\|^2$ that the SEA analysis already accounts for at every round. Consequently, the SEA regret bound applies unchanged to gradual drift (\Cref{ex:gradual_drift}), which is absorbed directly into $\bar{\Sigma}_T^2$ without requiring a modified analysis.

The mission regimes in Examples~\ref{ex:single_regime}--\ref{ex:rapid_switching} can be characterized precisely: single regime corresponds to $J = O(1)$ and $\bar{\Sigma}_T^2 \approx 0$; gradual drift to $J = O(\sqrt{T})$ with slowly growing $\bar{\Sigma}_T^2$; and rapid switching to $J = \Theta(T)$ with large $\bar{\Sigma}_T^2$. Applying the SEA framework to our problem requires (i) convexity of $f(\cdot,\xi)$ for $\xi$, (ii) bounded gradients, and (iii) smoothness. We verify below that our loss function satisfies these requirements.

\begin{lemma}[Linearity and Smoothness]
\label{lem:linear_smooth}
For any $\bm{\ell}_t\in[0,1]^K$, the loss $\ell_t(\bm{x})=\langle\bm{x},\bm{\ell}_t\rangle$ is linear (hence convex) on $\Delta^K$ with $\|\bm{\ell}_t\| \le \sqrt{K}$. Moreover, $\ell_t$ is $0$-smooth, and the expected loss $F_t(\bm{x})=\langle\bm{x},\bm{\mu}_t\rangle$ inherits these properties.
\end{lemma}

\begin{proof}
Since $f(\bm{x}, \xi) = \langle\bm{x},\bm{\ell}(\xi)\rangle$ is linear in $\bm{x}$, it is $0$-smooth~\cite{sachs2022between}. Because each coordinate satisfies $L(\pi_i,\xi_t)\in[0,1]$, we have $\|\bm{\ell}_t\| \le \sqrt{K}$. Taking expectations preserves linearity: $F_t(\bm{x})=\mathbb{E}_{\xi\sim D_{s_t}}[f(\bm{x},\xi)] =\langle\bm{x},\bm{\mu}_t\rangle$.
\end{proof}
\section{Multi-Regime Policy Selection for UAV Search Missions}
\label{sec:multiregime} 

As the UAV navigates, radar statistics evolve due to changing scenes. The policy selection task must exploit the best policy for the current scene while remaining robust to scene transitions. We model this as online convex optimization on the simplex $\Delta^K$ and employ the SEA framework, which captures both the stochastic variance $\bar{\sigma}_T^2$ (in-scene noise) and the adversarial variation $\bar{\Sigma}_T^2$ (inter-scene drift).

\subsection{Optimistic FTRL with Adaptive Learning Rate}
\label{sec:oftrl_algorithm}

We refer to the resulting selector as \textsc{SEArch}, an SEA-based online policy selector for UAV radar search. \textsc{SEArch} uses OFTRL~\cite{rakhlin2013predictable} with an adaptive learning rate~\cite{sachs2022between} that responds to observed loss variability. At round $t$, the algorithm forms an optimistic prediction $\bm{M}_t$ of the next loss, then selects the policy weights by minimizing a regularized cumulative loss~\cite[Section 2.1]{rakhlin2013predictable} (see also~\cite[Section 7.12]{orabona2019modern}):
\begin{equation}
\label{eq:oftrl_update}
\bm{x}_t = \argmin_{\bm{x}\in\Delta^K} \left\{ \left\langle\bm{x}, \bm{M}_t + \bm{L}_{t-1}\right\rangle + \frac{1}{\eta_t}\|\bm{x}\|^2 \right\},
\end{equation}
where $\bm{L}_{t-1} = \sum_{r=1}^{t-1} \bm{\ell}_r$ is the cumulative loss up to round $t-1$. The first term drives $\bm{x}_t$ toward policies with low cumulative loss. For instance, suppose the UAV has three policies: high-frequency radar ($\pi_1$), mid-frequency ($\pi_2$), and low-frequency ($\pi_3$). If recent observations strongly favor $\pi_1$ due to favorable clutter conditions, this term would drive $\bm{x}_t$ toward $[1, 0, 0]$. The regularization term $\frac{1}{\eta_t}\|\bm{x}\|^2$ (with time-varying learning rate $\eta_t$) prevents aggressive updates and ensures that $\bm{x}_t \in \Delta^K$.

Following the SEA framework, we set $\bm{M}_t = \bm{\ell}_{t-1}$ as the prediction of the next loss. Within a stochastic scene, $\bm{\ell}_{t-1}$ provides an unbiased estimate of the current mean $\bm{\mu}_t$ (since $\mathbb{E}[\bm{\ell}_{t-1}] = \bm{\mu}_t$ when $s_t = s_{t-1}$), yielding accurate predictions. At scene transitions, the prediction becomes inaccurate, but the adaptive learning rate mechanism (described next) automatically detects this and reduces the learning rate in response.

\subsubsection{Adaptive Learning Rate}
The SEA learning rate~\cite[Thm.~5]{sachs2022between} adapts based on prediction errors:
\begin{equation}
\label{eq:eta_adaptive}
\eta_t = \frac{D^2}{\nu + \sum_{r=1}^{t-1} \eta_r \cdot \|\bm{\ell}_r - \bm{M}_r\|^2},
\end{equation}
where $D$ is the diameter of the simplex $\Delta^K$ ($D = \sqrt{2}$ for the probability simplex) and $\nu > 0$ is a tuning parameter that ensures $\eta_t$ remains bounded. 

When predictions are accurate (stable scenes), the denominator grows slowly and $\eta_t$ remains large, enabling aggressive adaptation that exploits the stable structure. When losses are unpredictable (scene transitions or high noise), the denominator grows rapidly, $\eta_t$ decreases, and updates become conservative to avoid overreacting to transient fluctuations.

\subsubsection{Algorithm Summary}
\Cref{alg:oftrl_sea} summarizes the complete procedure. At each round, the UAV (1) forms an optimistic prediction using the previous loss (line 3), (2) computes the adaptive learning rate (line 4) and policy weights (line 5), (3) samples and executes a policy according to these weights (line 6), and (4) observes the full loss vector $\bm{\ell}_t$ (line 7). The algorithm requires no prior knowledge of the number of scenes $J$, the switch times $\{\tau^{(j)}\}$, or the variance/variation parameters $\bar{\sigma}_T^2, \bar{\Sigma}_T^2$: it adapts purely based on observed prediction errors.

\begin{algorithm}[t]
\caption{\textsc{SEArch}}
\label{alg:oftrl_sea}
\begin{algorithmic}[1]
\Require Policy set $\Pi = \{\pi_1, \ldots, \pi_K\}$, horizon $T$, tuning parameter $\nu > 0$, simplex diameter $D = \sqrt{2}$
\State \textbf{Initialize:} $\bm{L}_0 \gets \bm{0} \in \mathbb{R}^K$, $\bm{M}_1 \gets \bm{0}$, $\bm{\ell}_0 \gets \bm{0}$

\For{$t = 1, 2, \ldots, T$}
    \State Set optimistic prediction: $\bm{M}_t \gets \bm{\ell}_{t-1}$
    \State Compute adaptive learning rate via \Cref{eq:eta_adaptive}
    \State Compute policy weights via \Cref{eq:oftrl_update}
    \State Sample policy index $i \sim \bm{x}_t$ and execute $\pi_i$
    \State Observe full loss vector $\bm{\ell}_t \in [0,1]^K$
    \State Update cumulative loss: $\bm{L}_t \gets \bm{L}_{t-1} + \bm{\ell}_t$
\EndFor
\end{algorithmic}
\end{algorithm}

\subsection{Performance Guarantee}
\label{sec:performance_guarantee}

We now state the regret guarantee for \Cref{alg:oftrl_sea} in our UAV radar setting. The loss functions are linear (hence convex and $L$-smooth with $L=0$), with domain diameter $D = \sqrt{2}$ and gradient bound $G = \sqrt{K}$. Since $K$ is a fixed mission-level parameter, we treat $G$ as a constant and absorb it into the big-$O$ notation in what follows. By~\cite[Thm.~5]{sachs2022between}, \Cref{alg:oftrl_sea} with $\nu > 0$ achieves the following expected regret bound for any fixed comparator $\bm{x}^* \in \Delta^K$:
\begin{equation}
\label{eq:sea_regret_bound}
\mathbb{E}[\mathrm{Reg}_T] \leq O\left((\bar{\sigma}_T + \bar{\Sigma}_T)\sqrt{T}\right).
\end{equation}
For missions with $J$ scene transitions, we specialize this bound by relating $\bar{\Sigma}_T$ to the number of scene switches.

\begin{corollary}[Regret in Terms of Scene Transitions]
\label{cor:switch_regret}
Suppose the mission consists of $J$ contiguous stationary segments (\Cref{def:segments}), with mean loss vector $\bm{\mu}^{(j)}$ in segment $I^{(j)}$ as in \Cref{eq:mu_segment}. Then, under piecewise stationarity,
\begin{equation}
\label{eq:switch_bound}
\mathbb{E}[\mathrm{Reg}_T] \leq O\left(\bar{\sigma}_T \sqrt{T} + \sqrt{J}\right).
\end{equation}
\end{corollary}

\begin{proof}
By \Cref{def:variance_variation} and piecewise stationarity (\Cref{def:segments}), the adversarial variation satisfies
\[
\bar{\Sigma}_T^2 = \frac{1}{T}\sum_{t=2}^T \|\bm{\mu}_t - \bm{\mu}_{t-1}\|^2.
\]
Since $\bm{\mu}_t = \bm{\mu}^{(j)}$ is constant within segment $I^{(j)}$, the sum collapses to the $J-1$ switch times $\{\tau^{(j)}\}_{j=2}^{J}$, giving
\[
\bar{\Sigma}_T^2 = \frac{1}{T}\sum_{j=2}^{J} \|\bm{\mu}^{(j)} - \bm{\mu}^{(j-1)}\|^2 \leq \frac{4 (J-1) G^2}{T},
\]
where we used $\|\bm{\mu}^{(j)} - \bm{\mu}^{(j-1)}\|^2 \leq 4G^2$ via the triangle inequality. Taking square roots, $\bar{\Sigma}_T \leq 2G\sqrt{(J-1)/T}$, so $\bar{\Sigma}_T\sqrt{T} \leq 2G\sqrt{J-1} = O(\sqrt{J})$ after absorbing $G$ into the constant. Substituting into \Cref{eq:sea_regret_bound} yields the bound.
\end{proof}

While this algorithm handles moderate regime changes, it accumulates regret over $[T]$. For missions with frequent transitions (as in Example~\ref{ex:rapid_switching}), this global accumulation is suboptimal: the algorithm remains influenced by irrelevant past scenes. The next section addresses this limitation.
\section{Windowed Policy Selection for Rapidly Changing Missions}
\label{sec:windowed} 

Missions with frequent scene transitions require algorithms that adapt within each segment rather than relying on full historical information. \Cref{alg:oftrl_sea} achieves the regret bound in \Cref{eq:sea_regret_bound}, which scales with $\bar{\sigma}_T$ and $\bar{\Sigma}_T$ over the full horizon. This is suitable for missions with infrequent transitions ($J = O(1)$) or gradual drift ($J = O(\sqrt{T})$), but is suboptimal under rapid switching because the algorithm accumulates losses from irrelevant past scenes. We therefore introduce a \emph{windowed} OFTRL variant that periodically restarts both the loss accumulation in \Cref{eq:oftrl_update} and the adaptive learning rate in \Cref{eq:eta_adaptive}.

\subsection{Windowed OFTRL Algorithm}
\label{sec:windowed_oftrl} 

The windowed algorithm performs a restart every $w$ rounds, where $w$ is a user-specified window length. Let $b_t = \lfloor (t-1)/w \rfloor \cdot w + 1$ denote the start of the current window containing round $t$. For instance, with $w=10$, rounds $t=1,\ldots,10$ have $b_t=1$, rounds $t=11,\ldots,20$ have $b_t=11$, and so on. At each round $t$, the algorithm solves
\begin{equation}
\label{eq:ftrl_window_update}
\bm{x}_t
=
\argmin_{\bm{x}\in\Delta^K}
\left\{
\left\langle\bm{x},\,
\bm{M}_t +
\sum_{r=b_t}^{t-1} \bm{\ell}_r
\right\rangle
+
\frac{1}{\eta_t}\|\bm{x}\|^2
\right\},
\end{equation}
where the cumulative loss is restricted to the current window $[b_t, t-1]$. Similarly, the learning rate is computed using only losses within the current window:
\begin{equation}
\label{eq:eta_window}
\eta_t = \frac{D^2}{\nu + \sum_{r=b_t}^{t-1} \eta_r \cdot  \|\bm{\ell}_r - \bm{M}_r\|^2}.
\end{equation}
At $t = k \cdot w + 1$ for $k = 1, 2, \ldots$, the algorithm performs a restart: all accumulated losses and predictions are discarded, and optimization resumes with $\bm{M}_t = \bm{0}$ and an empty cumulative loss.

\subsubsection{Window-Size Selection}
$w$ trades off stability and adaptivity. Large windows ($w \approx T$) recover the global algorithm in \Cref{alg:oftrl_sea}, yielding low regret on long segments but slow adaptation. Small windows ($w \approx \log T$) enable rapid adaptation but may exhibit high variance within segments. In practice, $w$ should be chosen based on typical segment lengths.

\begin{algorithm}[t]
\caption{\textsc{W-SEArch}}
\label{alg:windowed_oftrl}
\begin{algorithmic}[1]
\Require Policy set $\Pi = \{\pi_1, \ldots, \pi_K\}$, horizon $T$, window length $w$, tuning parameter $\nu > 0$, simplex diameter $D = \sqrt{2}$
\State \textbf{Initialize:} $\bm{M}_1 \gets \bm{0}$, $\bm{\ell}_0 \gets \bm{0}$, $b_1 \gets 1$
\For{$t = 1, 2, \ldots, T$}
    \If{$(t-1) \bmod w = 0$ and $t > 1$}
        \State \textbf{Hard restart:} $\bm{M}_t \gets \bm{0}$, $b_t \gets t$
    \Else
        \State $b_t \gets b_{t-1}$
    \EndIf
    \State Set optimistic prediction: $\bm{M}_t \gets \bm{\ell}_{t-1}$
    \State Compute adaptive learning rate via \Cref{eq:eta_window}
    \State Compute policy weights via \Cref{eq:ftrl_window_update}
    \State Sample policy index $i \sim \bm{x}_t$ and execute $\pi_i$
    \State Observe full loss vector $\bm{\ell}_t \in [0,1]^K$
\EndFor
\end{algorithmic}
\end{algorithm}

\subsection{Regret Guarantee}
\label{sec:windowed_guarantee}

We now state the adaptive regret guarantee for \Cref{alg:windowed_oftrl}. Recall from \Cref{sec:problem_formulation} that adaptive regret measures performance over all contiguous intervals of a given length, allowing the comparator to change across windows.

\begin{theorem}[Windowed Adaptive Regret]
\label{thm:windowed_regret}
Let $w \leq T$ be the window length. For any contiguous interval $I = [s, s+w-1] \subseteq [T]$ with $|I| = w$, \Cref{alg:windowed_oftrl} with $D = \sqrt{2}$, $G = \sqrt{K}$, and $L = 0$ achieves
\begin{equation}
\label{eq:windowed_regret}
\mathbb{E}[\mathrm{Reg}_I] \leq O\left(\bar{\sigma}_I \sqrt{w} + \sqrt{J_I}\right),
\end{equation}
where $\bar{\sigma}_I$ is the stochastic standard deviation restricted to $I$ (\Cref{def:variance_variation} with summation over $t \in I$), and $J_I$ is the number of scene transitions within $I$. If at most one scene transition occurs per window, i.e., $J_I \leq 1$ for all $I$, the per-interval bound simplifies to
\begin{equation}
\label{eq:windowed_regret_simplified}
\mathbb{E}[\mathrm{Reg}_I] \leq O\left(\bar{\sigma}_I \sqrt{w}\right).
\end{equation}
\end{theorem}

\begin{proof}
Let $W_k = [kw+1, (k+1)w]$ denote the $k$-th aligned window, at which \Cref{alg:windowed_oftrl} performs a restart. Within each $W_k$, the algorithm reduces to a fresh instance of \Cref{alg:oftrl_sea} applied to $w$ rounds with empty history. Applying \Cref{cor:switch_regret} with horizon $w$ and $J_{W_k}$ transitions yields
\[
\mathbb{E}[\mathrm{Reg}_{W_k}] \leq O\left(\bar{\sigma}_{W_k}\sqrt{w} + \sqrt{J_{W_k}}\right).
\]

Now consider any contiguous interval $I \subseteq [T]$ of length $w$. Since $|I| = w$ and the aligned windows also have length $w$, $I$ intersects at most two consecutive aligned windows $W_k$ and $W_{k+1}$. Applying the per-window bound to each and summing,
\[
\mathbb{E}[\mathrm{Reg}_I] \leq \mathbb{E}[\mathrm{Reg}_{W_k}] + \mathbb{E}[\mathrm{Reg}_{W_{k+1}}] 
\leq O\left(\bar{\sigma}_I\sqrt{w} + \sqrt{J_I}\right),
\]
where the absolute constant induced by summing two aligned-window bounds is absorbed into the big-$O$. Substituting $J_I \leq 1$ yields \Cref{eq:windowed_regret_simplified}. By contrast, the global bound in \Cref{eq:switch_bound} accumulates regret over the full horizon with total $J \leq T/w$ transitions, giving $O(\bar{\sigma}_T\sqrt{T} + \sqrt{T/w})$, whose switching penalty grows with $\sqrt{T/w}$. The windowed bound pays only a constant switching penalty per window and is therefore strictly tighter whenever $w \ll T$.
\end{proof}

\subsection{Computational and Memory Footprint at the Edge}
\label{sec:footprint}

Both \textsc{SEArch} and \textsc{W-SEArch} are deliberately lightweight, with no online training and no accelerator. Each round forms the prediction $\bm{M}_t \gets \bm{\ell}_{t-1}$ and the adaptive rate in $O(K)$, and solves \Cref{eq:oftrl_update} (resp.\ \Cref{eq:ftrl_window_update}) by a simplex projection in $O(K\log K)$, dominated by the policy count $K$ ($K \leq 16$ here). The state is only $O(K)$ and independent of the horizon $T$, since \Cref{alg:windowed_oftrl} keeps running within-window sums rather than buffering the last $w$ losses; windowing thus removes \emph{stale} information rather than reducing the byte count. Finally, by the full-information assumption (A1), a single measurement $\xi_t$ evaluates all $K$ policies locally, so the online layer adds no extra sensing or air-interface cost beyond detection.
\section{Performance Evaluation}
\label{sec:performance}

We evaluate the proposed algorithms through experiments across various scenarios and validate the theoretical regret guarantees established in \Cref{sec:multiregime,sec:windowed}.

\subsection{Experimental Setup}

We simulate the online policy selection problem with $K = 4$. \Cref{tab:policy_config} lists the mean loss of each policy across operating conditions. We set the horizon $T = 150$ with piecewise-stationary losses as in \Cref{def:segments}. Within each segment, losses are drawn from $\mathcal{N}(\bm{\mu}_t, \sigma^2 \mathbf{I})$, where $\bm{\mu}_t \in [0,1]^K$ is the segment mean and $\sigma = 0.15$. To model realistic operational transitions, policy switches occur gradually over a 10-step interpolation period. The simplex diameter is $D = \sqrt{2}$ and the gradient bound is $G = \sqrt{K} = 2$. We implement OFTRL (\Cref{alg:oftrl_sea}) with the SEA learning rate and its windowed variant (\Cref{alg:windowed_oftrl}) with $w = 30$, using $\nu = 1$ and averaging over $10$ runs. We compare against UCB~\cite{auer2002finite} with $c = 2.0$, Exp3~\cite{auer2002nonstochastic} with $\eta = 0.1$, and OMD with $\eta_0 = 0.5$.

\begin{table}[t]
\centering
\caption{Policy loss profiles across operating conditions.}
\label{tab:policy_config}
\small
\begin{tabular}{@{}lcccc@{}}
\toprule
\textbf{Policy} & \textbf{Scene 1} & \textbf{Scene 2} & \textbf{Scene 3} & \textbf{Scene 4} \\
\midrule
$\pi_1$ & 0.10 & 0.30 & 0.40 & 0.35 \\
$\pi_2$ & 0.33 & 0.12 & 0.38 & 0.40 \\
$\pi_3$ & 0.40 & 0.35 & 0.15 & 0.30 \\
$\pi_4$ & 0.35 & 0.40 & 0.30 & 0.15 \\
\bottomrule
\end{tabular}
\end{table}
\begin{figure*}
    \centering
    \includegraphics[width=\textwidth]{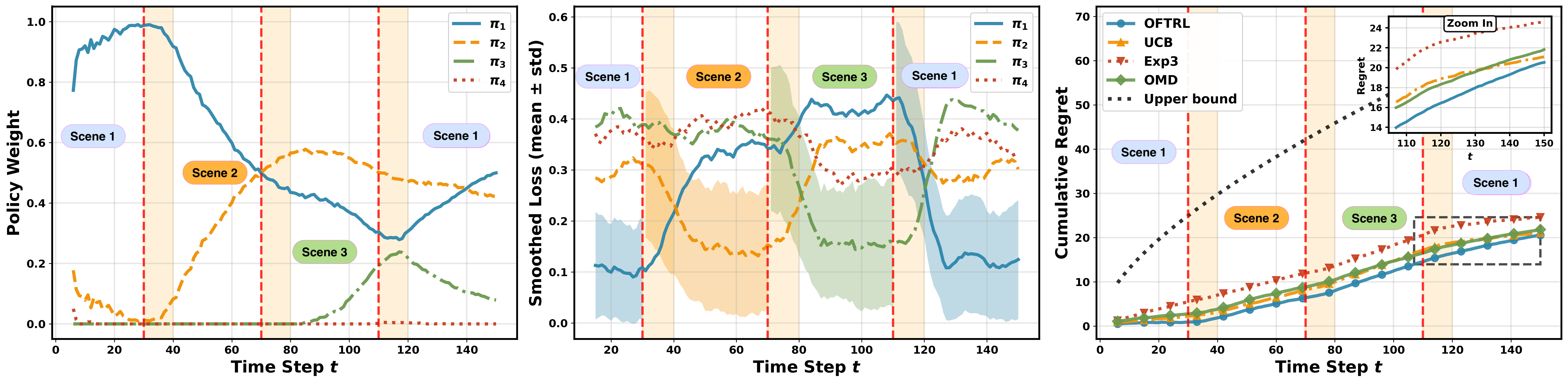}
    \caption{Performance comparison with 3 distribution shifts. Following the scenario in \Cref{fig:two_policy_scene_shift}, a gradual drift from one scene to the next is allowed (yellow shaded regions) after each regime switch (red lines). The shaded bands in Panel (b) show the variance of the optimal policy $\pi^*$ during each scene. The experiment demonstrates OFTRL's reallocation of policy weight through transitions: when $\pi_1$ is optimal during Scene 1 (maintaining the minimum mean loss of $0.1$), OFTRL concentrates weight on $\pi_1$ as shown in Panel (a). This adaptation is validated by the regret measurement in Panel (c), which remains sublinear over the horizon $T$ and stays below the theoretical upper bound in \Cref{eq:sea_regret_bound}.}
    \label{fig:exp1_sota}
\end{figure*}

\subsection{Comparison with State-of-the-Art Methods}

We evaluate OFTRL on a scenario with $4$ segments, where $\pi^*$ cycles through $\pi_1 \to \pi_2 \to \pi_3 \to \pi_1$ with shifts at $t = 30, 70, 110$. \Cref{fig:exp1_sota} shows that OFTRL concentrates weight on the optimal policy within $10$ steps of each shift, adapting smoothly during the gradual transition periods. This rapid adaptation is driven by OFTRL's adaptive learning rate and optimistic predictions, which together provide a proactive response to regime changes without manual tuning. The algorithm achieves sublinear regret below the theoretical bound from \Cref{cor:switch_regret}, with a $2\%$ improvement over UCB, $17\%$ over Exp3, and $6\%$ over OMD.

\subsection{Adaptation to Different Operating Regimes}

\begin{figure*}
    \centering
    \begin{subfigure}[b]{\textwidth}
        \includegraphics[width=\textwidth]{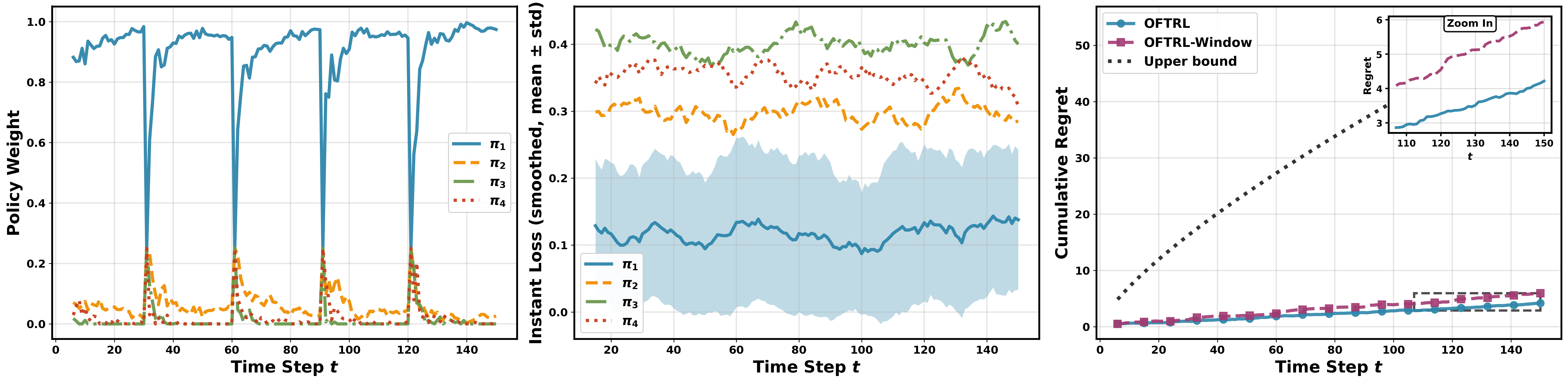}
        \caption{Single-regime dominance ($J = 0$).}
        \label{fig:exp2_single}
    \end{subfigure}
    \hfill
    \begin{subfigure}[b]{\textwidth}
        \includegraphics[width=\textwidth]{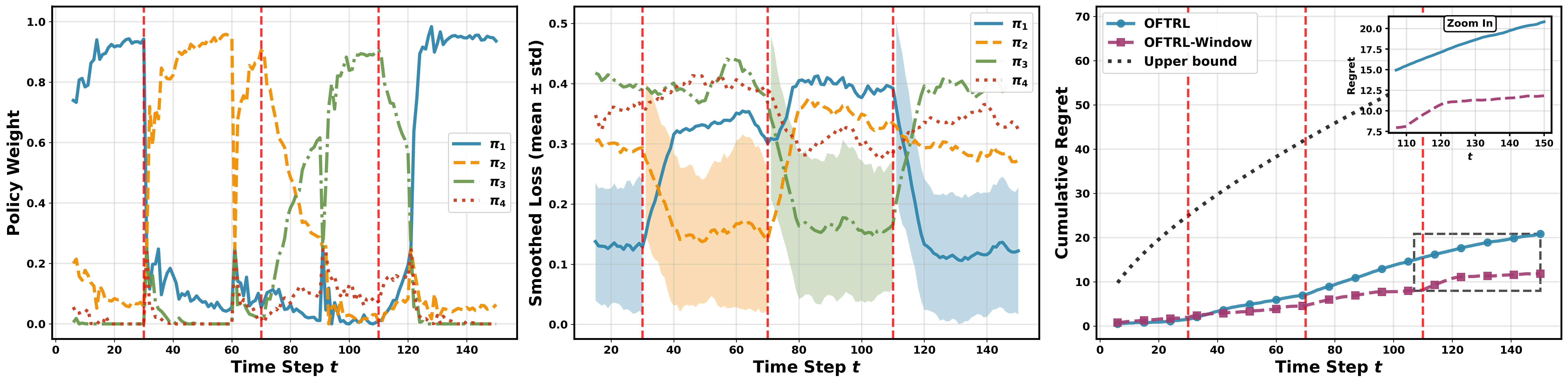}
        \caption{Rapid regime switching ($J = 3$).}
        \label{fig:exp2_three}
    \end{subfigure}
    \hfill
    \begin{subfigure}[b]{\textwidth}
        \includegraphics[width=\textwidth]{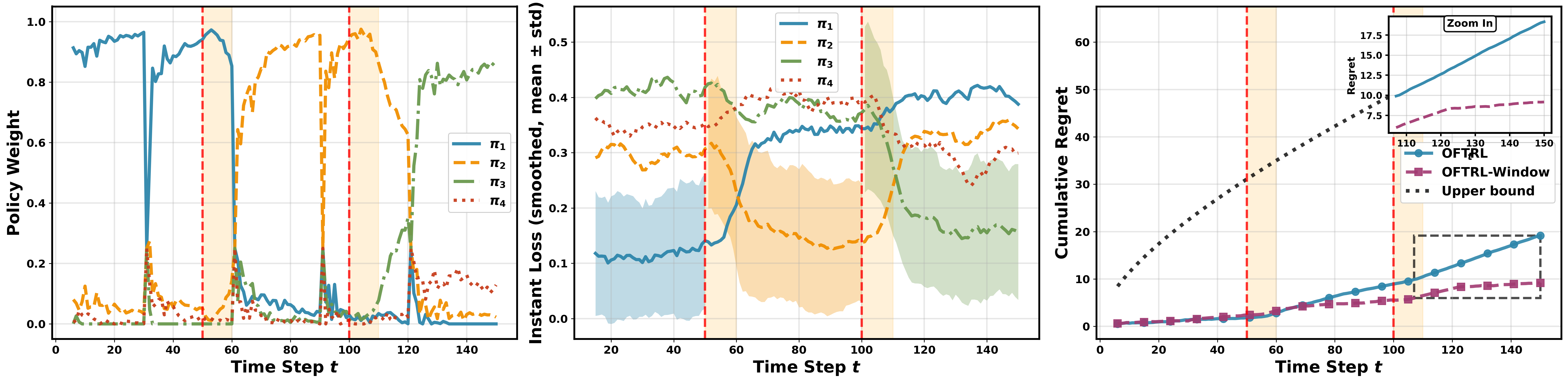}
        \caption{Gradual regime drift ($J = 2$).}
        \label{fig:exp2_gradual}
    \end{subfigure}
    \caption{Performance on different operating regimes. The windowed algorithm adapts faster in non-stationary regimes, while both algorithms achieve comparable performance in stable conditions.}
    \label{fig:exp2_system}
\end{figure*}

We next validate the algorithms on the different regimes described in \Cref{sec:regime_examples}. \Cref{fig:exp2_single} shows that both algorithms achieve sublinear regret under a stationary horizon, with comparable performance in stable environments. \Cref{fig:exp2_three} demonstrates the windowed algorithm's advantage through faster policy reallocation at regime boundaries, achieving a $60\%$ reduction in final regret compared to OFTRL. \Cref{fig:exp2_gradual} shows that the windowed variant successfully tracks the slowly varying optimal policy, yielding a $72\%$ reduction in regret compared to OFTRL.

\subsection{Robustness to Varying Non-Stationarity}

\begin{figure*}
    \centering
    \begin{subfigure}[b]{\textwidth}
        \includegraphics[width=\textwidth]{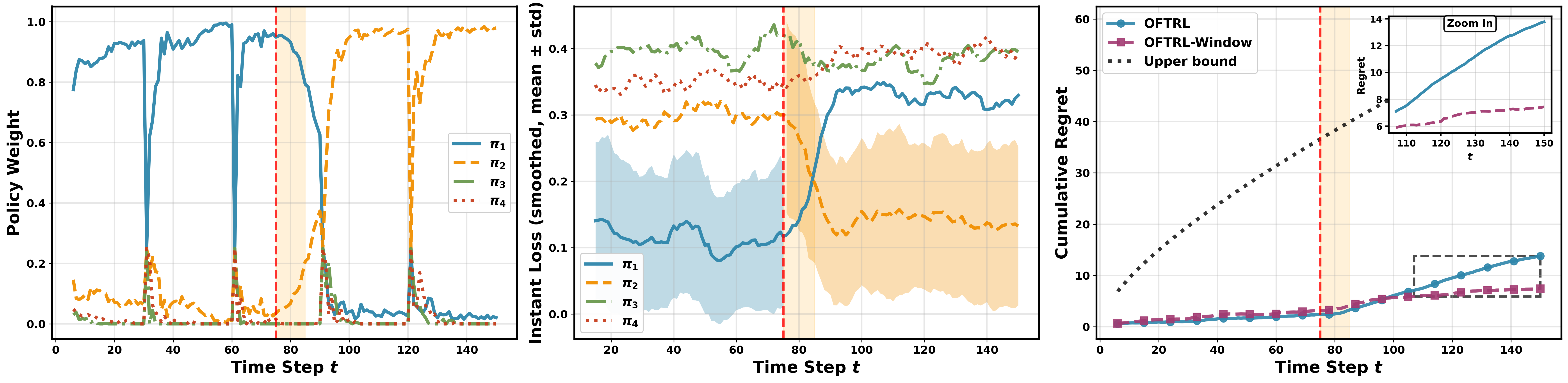}
        \caption{Single transition ($J = 1$).}
        \label{fig:exp3_one}
    \end{subfigure}
    \hfill
    \begin{subfigure}[b]{\textwidth}
        \includegraphics[width=\textwidth]{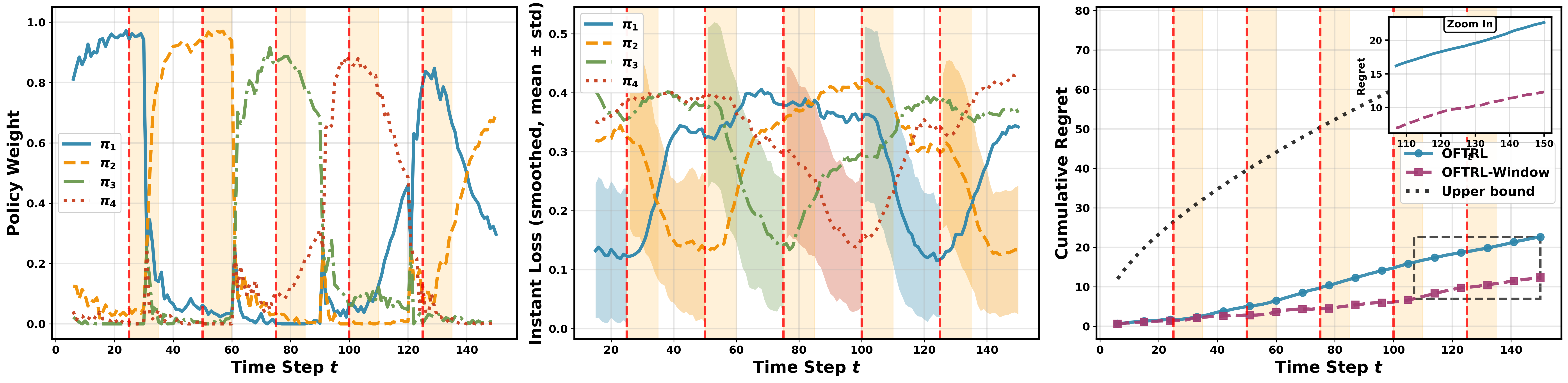}
        \caption{Five transitions ($J = 5$).}
        \label{fig:exp3_five}
    \end{subfigure}
    \hfill
    \begin{subfigure}[b]{\textwidth}
        \includegraphics[width=\textwidth]{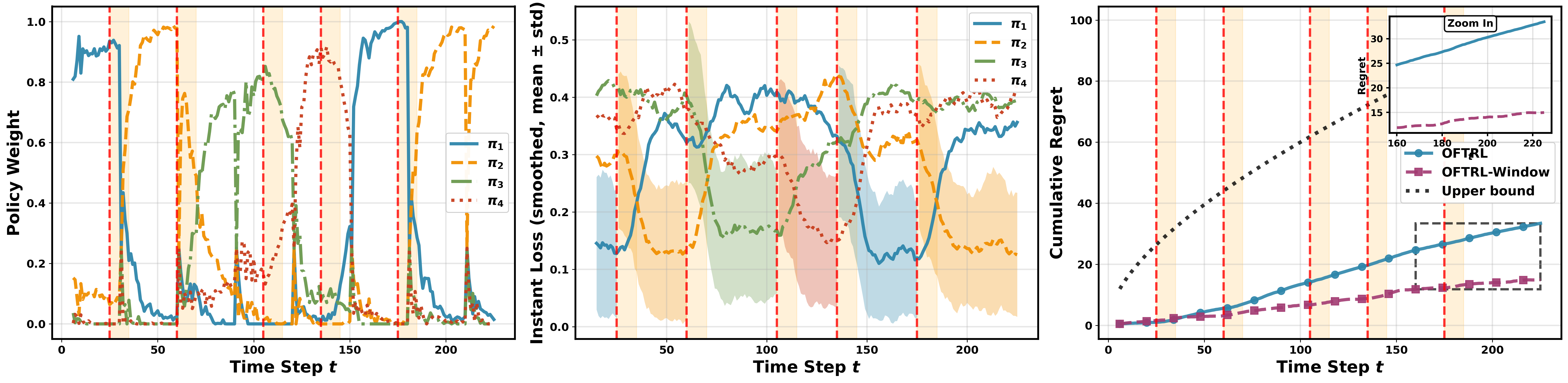}
        \caption{Five transitions with inconsistent durations ($J = 5$).}
        \label{fig:exp3_five_inconsistent}
    \end{subfigure}
    \caption{Performance under varying numbers of scene transitions. The windowed algorithm maintains substantially lower regret as the number of transitions increases, consistent with the theoretical $O(\sqrt{J})$ penalty from \Cref{cor:switch_regret}. Even with inconsistent segment durations ranging from $25$ to $50$ steps, the windowed approach adapts effectively.}
    \label{fig:exp3_switches}
\end{figure*}

\Cref{fig:exp3_switches} shows the windowed algorithm's robustness to transition frequency. With one shift (\Cref{fig:exp3_one}), the windowed algorithm achieves a $72\%$ reduction in regret compared to OFTRL. With five transitions (\Cref{fig:exp3_five}), it achieves a $57\%$ reduction. Importantly, \Cref{fig:exp3_five_inconsistent} demonstrates robustness to irregular switching: with $5$ transitions spaced over inconsistent intervals, the windowed algorithm achieves a $70\%$ reduction in regret compared to OFTRL. This validates that the windowed approach adapts effectively even when segment durations are unpredictable, as the restart mechanism operates independently of the underlying regime structure.

\begin{figure*}[t!]
    \centering
    \begin{subfigure}[b]{\textwidth}
        \includegraphics[width=\textwidth]{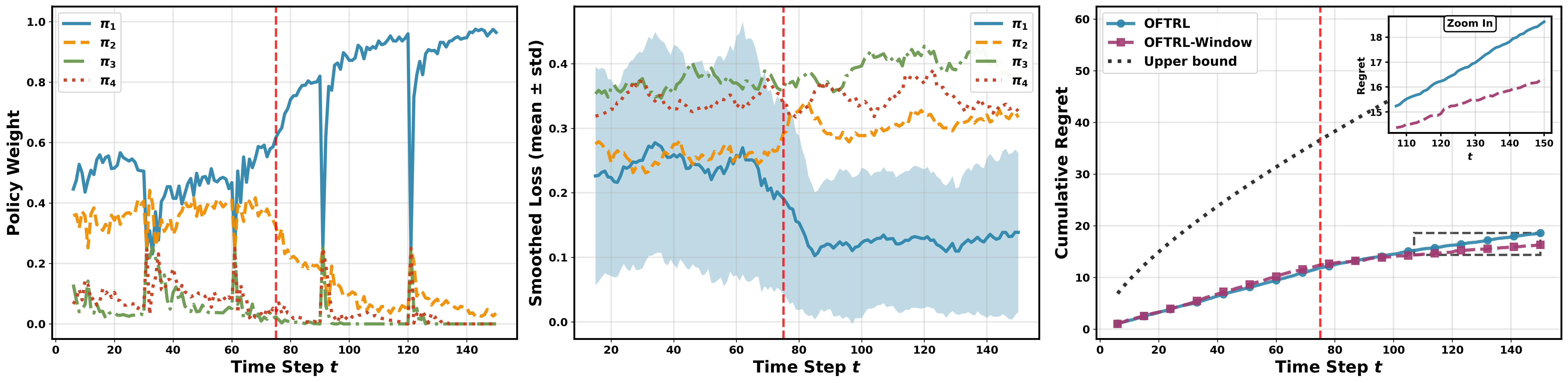}
        \caption{Non-stationary to stable transition.}
        \label{fig:exp3_hybrid_v2s}
    \end{subfigure}
    \caption{Performance under hybrid environments transitioning from non-stationary to stable conditions. The windowed algorithm demonstrates rapid adaptation when conditions stabilize.}
    \label{fig:exp3_nonstationarity}
\end{figure*}

\Cref{fig:exp3_nonstationarity} illustrates how limited memory affects regime transitions. When conditions shift to a stable regime (\Cref{fig:exp3_hybrid_v2s}), windowed OFTRL achieves a $6\%$ improvement in regret. This confirms that limited memory is beneficial for transitions \emph{into} predictable environments, where discarding noisy history enables rapid learning.

\subsection{Ablation Studies}

\begin{figure*}
    \centering
    \includegraphics[width=\textwidth]{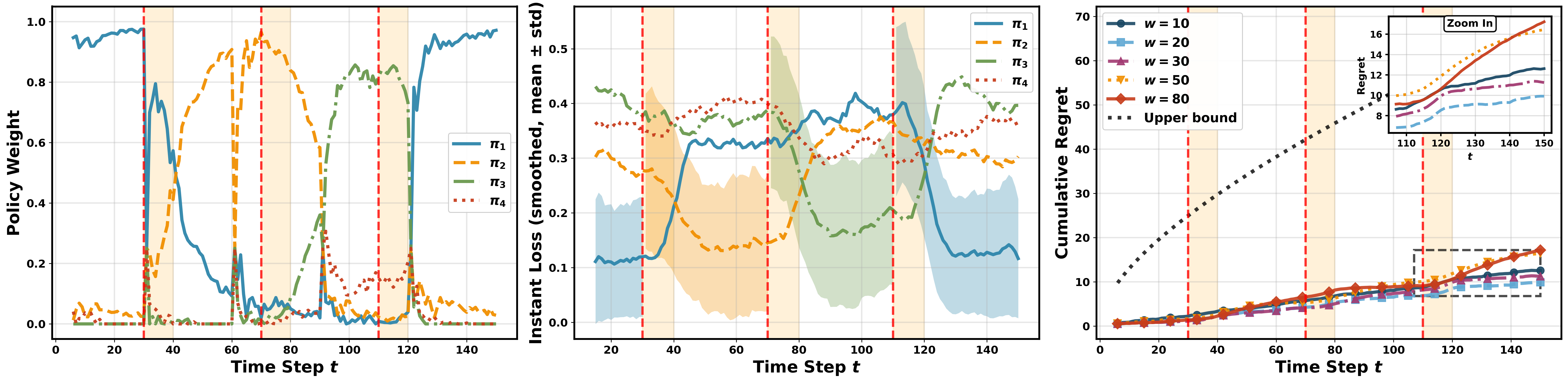}
    \caption{Ablation study on window size ($w \in \{10, 20, 30, 50, 80\}$). A window of $w=20$ provides the best balance for this $3$-switch scenario.}
    \label{fig:exp4_window}
\end{figure*}

\Cref{fig:exp4_window} shows the impact of window size in the $3$-switch scenario. Smaller windows adapt quickly but exhibit higher variance, whereas the optimal window size balances adaptation and stability. Medium windows show comparable performance, while larger windows provide excessive stability at the cost of slower adaptation. We also examine performance across segment durations with fixed $w = 30$, evaluating scenarios with segment lengths of $15$, $30$, and $50$ steps ($T = 150$, policy cycling through $\{\pi_1, \pi_2, \pi_3, \pi_4\}$).

\begin{table}[t]
\centering
\caption{Impact of segment duration on cumulative regret with $w = 30$. The windowed algorithm performs best when the segment length matches or exceeds the window size.}
\label{tab:ablation_regret}
\small
\begin{tabular}{@{}lccc@{}}
\toprule
\textbf{Segment config.} & \textbf{Length (steps)} & \textbf{Windowed regret} & \textbf{Switches ($J$)} \\
\midrule
Short & 15 & $14.96 \pm 0.83$ & 9 \\
Medium & 30 & $\phantom{0}6.46 \pm 1.26$ & 4 \\
Long & 50 & $\phantom{0}5.57 \pm 1.11$ & 2 \\
\bottomrule
\end{tabular}
\end{table}

\Cref{tab:ablation_regret} reveals a clear relationship between segment duration and performance. When segments ($15$ steps) are shorter than the window size ($30$ steps), frequent restarts limit the amount of learning within each segment, resulting in a regret of $14.96$ across $9$ switches. When the segment length matches the window size ($30$ steps), performance improves dramatically to a regret of $6.46$ with $4$ switches. When segments exceed the window size ($50$ steps), regret further decreases to $5.57$ with only $2$ switches. This validates a simple design rule: \emph{set $w$ to match the minimum expected segment duration}.

We next evaluate $4$ noise levels $\sigma \in \{0.05, 0.1, 0.2, 0.4\}$ on the $3$-transition scenario.

\begin{table}[t]
\centering
\caption{Noise sensitivity: cumulative regret across noise levels. All algorithms degrade with increased noise, but the windowed approach maintains its advantage, especially at lower noise levels.}
\label{tab:noise_sensitivity}
\small
\begin{tabular}{@{}lccccc@{}}
\toprule
\multirow{2}{*}{$\sigma$} & \multicolumn{4}{c}{\textbf{Regret $\mathrm{Reg}_T$}} & \textbf{Window} \\
\cmidrule(lr){2-5}
& \textbf{OFTRL} & \textbf{Window} & \textbf{UCB} & \textbf{Exp3} & \textbf{Improv. (\%)} \\
\midrule
0.05 & 18.82 & 10.99 & 19.55 & 22.49 & $+41.6$ \\
0.10 & 19.15 & 11.33 & 19.68 & 22.53 & $+40.8$ \\
0.20 & 21.19 & 16.04 & 26.18 & 25.64 & $+24.3$ \\
0.40 & 31.12 & 27.20 & 36.88 & 31.64 & $+12.6$ \\
\bottomrule
\end{tabular}
\end{table}

\Cref{tab:noise_sensitivity} shows that all algorithms degrade as noise increases (regret grows from $11$--$23$ at $\sigma = 0.05$ to $27$--$37$ at $\sigma = 0.4$), but windowed OFTRL consistently achieves the lowest regret. The windowed advantage diminishes with noise (from $41.6\%$ at $\sigma = 0.05$ to $12.6\%$ at $\sigma = 0.4$), reflecting the fact that higher noise requires more observations per segment, which reduces the benefit of discarding history.

\textbf{Number of Policies $K$.}
We evaluate scalability with $K \in \{4, 8, 16\}$ in the three-switch scenario. \Cref{tab:ablation_k_policies_regret} shows performance across policy set sizes.

\begin{table}[t]
\centering
\caption{Effect of the number of policies $K$ on final regret ($T=150$, $J=3$). Both algorithms scale well, with windowed OFTRL maintaining an advantage as the action space grows.}
\label{tab:ablation_k_policies_regret}
\small
\begin{tabular}{@{}lccc@{}}
\toprule
$K$ & \textbf{$\mathrm{Reg}_T$ (OFTRL)} & \textbf{$\mathrm{Reg}_T$ (Window)} & \textbf{Upper bound} \\
\midrule
4  & $2.31 \pm 0.85$ & $3.66 \pm 1.20$ & 68.84 \\
8  & $5.17 \pm 0.56$ & $5.00 \pm 0.57$ & 77.18 \\
16 & $8.92 \pm 0.48$ & $6.26 \pm 0.74$ & 87.63 \\
\bottomrule
\end{tabular}
\end{table}

Both algorithms maintain robust performance as $K$ increases. For $K=4$, OFTRL achieves lower regret, which suggests that with few policies the overhead of frequent restarts outweighs the adaptation benefits. However, as the action space grows to $K=8$ and $K=16$, the windowed variant achieves superior performance.
\section{Conclusion}
\label{sec:conclusion}
We presented a framework for adaptive policy selection in UAV-based search under non-stationary radar conditions. Using the SEA framework, which accounts jointly for in-scene noise and inter-scene shifts, we developed two algorithms. The first, \textsc{SEArch}, is an optimistic FTRL selector achieving regret $O(\bar{\sigma}_T \sqrt{T} + \sqrt{J})$, where $J$ is the number of scene transitions. The second, \textsc{W-SEArch}, is a windowed variant achieving adaptive regret $O(\bar{\sigma}_I \sqrt{w})$ on any contiguous interval of length $w$. Both run entirely onboard at $O(K\log K)$ per-round cost with $O(K)$ state independent of the horizon, and consume no sensing or communication resources beyond detection, making them well suited to a resource-constrained aerial edge node. Experiments validate the theoretical predictions, with \textsc{W-SEArch} delivering $12\%$ to $35\%$ improvements as the transition frequency increases.

\subsection{Future Directions}
\textbf{Adaptive Window Selection.} Our fixed window size $w$ requires manual tuning based on the expected segment durations. Developing principled methods for online window adaptation~\cite{mhaisen25a, mhaisen2026partially} would eliminate this requirement and improve robustness to varying scene characteristics.

\textbf{Joint Trajectory and Policy Optimization.} We treat the UAV trajectory as given, but jointly optimizing the sensing trajectory and the processing policy could yield further gains. Coupling planning with policy selection~\cite{khial2025eswa, khial2025drone} would enable fully autonomous adaptive sensing systems that coordinate both where to sense and how to process observations.

\section*{Acknowledgment}
Research reported in this publication was supported by the Qatar Research Development and Innovation Council ARG01-0527-230356. The content is solely the responsibility of the authors and does not necessarily represent the official views of Qatar Research Development and Innovation Council.

\bibliographystyle{IEEEtran}
\bibliography{ref}

\end{document}